\documentclass[journal,doublecolumn]{IEEEtran}

\usepackage{amsmath,amssymb,amsthm}
\usepackage{graphicx}
\usepackage{pgfplots}
\pgfplotsset{compat=1.18}

\usepackage[style=ieee, backend=biber]{biblatex}
\usepackage{hyperref}
\hypersetup{
 colorlinks,
 linkcolor={blue},
 citecolor={blue}}

\newcommand{\TPC}{\mathrm{TPC}}
\newcommand{\rank}{\mathrm{rank}}

\newcommand{\mltset}[1]{\{\!\{#1\}\!\}}
\newcommand{\ceil}[1]{ {\lceil#1\rceil}}
\newcommand{\floor}[1]{ {\lfloor#1\rfloor}}

\newcommand{\F}{\mathbb{F}_2} 
\newcommand{\calU}{\mathcal{U}} 
\newcommand{\calY}{\mathcal{Y}} 

\newtheorem{theorem}{Theorem}[section]
\newtheorem{lemma}[theorem]{Lemma}

\newtheorem{definition}[theorem]{Definition}

\title{Achieving Torn-Paper Channel Capacity with Successive Revelation}

\author{Rui Xu, Le Wang

\thanks{Rui Xu is with KTH Royal Institute of Technology, Karolinska Institute, and Stockholm University (e-mail: rxu@kth.se). Le Wang is with School of Electrical Engineering and Computer Science, KTH Royal Institute of Technology, 10044 Stockholm, Sweden (email: le6@kth.se). }}

\begin{document}

\maketitle

\begin{abstract}
The torn-paper channel independently cuts a binary codeword at its internal boundaries and outputs the resulting oriented fragments as an unordered multiset. We consider the critical regime $p_N\log N\to\alpha$, in which the channel capacity is $e^{-\alpha}$. Existing coding schemes use a fixed-density pilot to localize fragments, creating a tradeoff between positional information and payload rate.
This paper introduces \textit{successive revelation}, which partitions the codeword into interleaved tracks and decodes them sequentially. Each recovered payload track becomes an additional positional reference for subsequent stages, allowing progressively shorter fragments to be localized.
We establish a finite-track achievable rate whose gap to capacity is $O(1/M)$ for $M$ tracks. Consequently, for every rate below the channel capacity, a finite number of tracks yields a sequence of deterministic codes with vanishing average decoding error probability.
\end{abstract}


\section{Introduction}

The torn-paper channel (TPC) models communication through stochastic fragmentation. A binary codeword of length $N$ is independently cut at each internal boundary with probability $p_N$, and the resulting oriented fragments are delivered to the receiver as an unordered multiset. Although originally motivated by DNA-based storage, the TPC is also a fundamental model of communication when fragmentation preserves local content but destroys global ordering.

Shomorony and Vahid~\cite{shomorony2020communicating,shomorony2021torn} introduced the TPC model and established the channel capacity
\[
  C_\TPC=e^{-\alpha}
\]
in the critical regime
\[
  p_N\log N\longrightarrow\alpha\in(0,\infty).
\]
Their converse and random-coding achievability characterize the fundamental limit, but existing coding schemes do not attain it. These schemes interleave payload bits with a synchronization \textit{pilot} string for positional reference. Increasing the reference density makes shorter fragments easier to localize, whereas decreasing it leaves more coordinates for payload. The resulting tradeoff keeps the best previously known structured rates strictly below capacity~\cite{shomorony2021torn,liu2026improved}.

This paper introduces \textit{successive revelation} (SR), a novel code ensemble that addresses this dilemma. The codeword coordinates are partitioned into $M$ periodic tracks. The first is a bootstrap track whose transmitted symbols are known to the decoder, while the remaining $M-1$ tracks carry separately encoded message segments. At stage $j$, the decoder uses the tracks already recovered to localize sufficiently long fragments and decode track $j$.
The newly recovered track then becomes part of the positional reference for subsequent stages. Thus, the reference density increases during decoding, the localization threshold decreases, and fragments that were unusable at an earlier stage can participant in the decoding at later stages.

The SR construction is \textit{capacity-achieving under average error}. That is, for every target rate $R<e^{-\alpha}$, there exists a finite track count $M$ for which the construction admits a sequence of deterministic public designs whose code rates converge to $R$ and whose average decoding error probabilities converge to zero. More precisely, for a fixed $M$, every rate below
\[
  R_M(\alpha)
  =
  \frac{1}{M}\sum_{j=1}^{M-1}
  \left(1+\frac{\alpha M}{j}\right)e^{-\alpha M/j}
\]
is achievable, and the gap to capacity is bounded as
\[
  0
  \leq
  e^{-\alpha}-R_M(\alpha)
  \leq
  \frac{(1+\alpha)e^{-\alpha}}{M}.
\]
Consequently, a finite value of $M$ suffices for every rate strictly below capacity.

The remainder of the paper is organized as follows.
Section~\ref{sec:related-work} provides an overview of related work, followed by Section~\ref{sec:preliminaries}, which introduces the notation, channel model, and error criterion. The main results are presented in Section~\ref{sec:results}. Section~\ref{sec:construction} then describes the SR encoder and decoder, and Section~\ref{sec:analysis} establishes the corresponding reliability and rate theorems. Finally, Section~\ref{sec:conclusion} concludes the paper.

\begin{figure*}[t]
  \centering
  \input{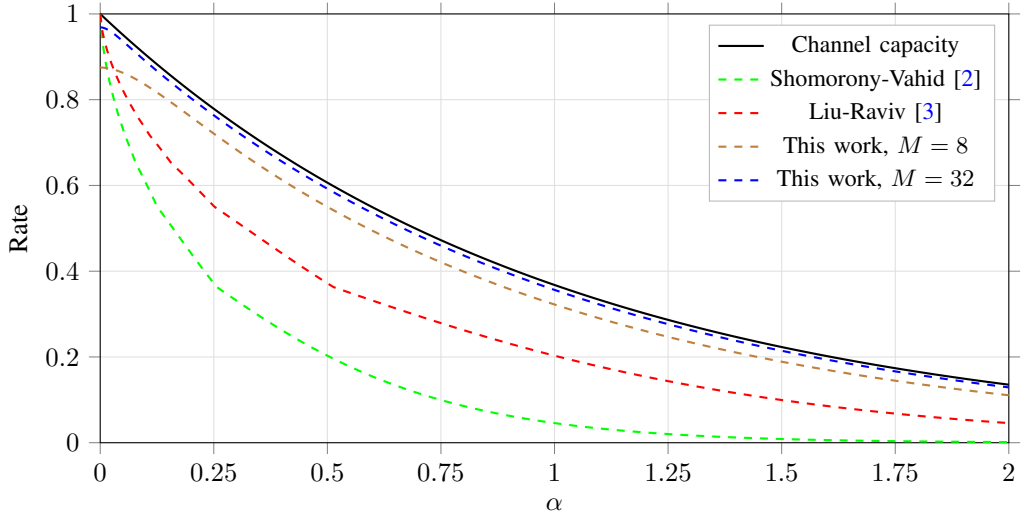}
  \caption{Capacity, interleaved-pilot rates, local-alignment rates, and finite-track SR rates.}
  \label{fig:rates}
\end{figure*}

\section{Related Work}
\label{sec:related-work}

Beyond the capacity characterization and achievability, Shomorony and Vahid proposed the \textit{interleaved-pilot} scheme~\cite{shomorony2020communicating,shomorony2021torn} for the torn-paper channel. For a fixed number $M\geq2$ of interleaved tracks, one track carries a positional reference based on a de Bruijn sequence, while the remaining tracks carry erasure-coded payloads.
The resulting achievable rate is
\[
  R_{\mathrm{IP}}
  =
  \left(1-\frac{1}{M}\right)
  (1+2M\alpha)e^{-2M\alpha}.
\]
Liu and Raviv subsequently introduced \textit{local alignment}, which makes additional short-fragment observations usable and achieves the improved fixed-$M$ rate~\cite{liu2026improved}
\[
  R_{\mathrm{LA}}
  =
  \left(1-\frac{1}{M}\right)
  (1+M\alpha)e^{-M\alpha}.
\]
Both schemes use a positional reference whose density is fixed before decoding. Successive revelation differs from the existing schemes in that every recovered payload track enlarges the reference available to later stages. In Fig.~\ref{fig:rates}, we compare the rate of interleaved-pilot and local-alignment, both optmized over~$M$, with the SR rates for $M=8$ and $M=32$.

Several extensions of the stochastic TPC incorporate additional channel noises. Ravi~\textit{et al.} studied fragment loss and the recovery of messages from incomplete collections of noisy fragments~\cite{ravi2021capacity,ravi2026recovering}. Walter~\textit{et al.} considered coding schemes for the noisy TPC~\cite{walter2026coding}.
A complementary line of research studies practical reassembly. Embedded and nested Varshamov--Tenengolts checks have been used to reassemble nonoverlapping random fragments~\cite{nassirpour2020embedded,nassirpour2023dna}.
Jiao \textit{et al.} developed a nested-hashing construction and evaluated its error probability and decoding complexity numerically~\cite{jiao2025efficient}. These works primarily target explicit fragment reassembly rather than the capacity of the stochastic TPC.

Adversarial fragmentation replaces random cut locations with worst-case cuts. Bar-Lev \textit{et al.} studied a model in which every nonterminal fragment must satisfy prescribed lower and upper length bounds~\cite{bar2023adversarial}. Wang \textit{et al.} introduced break-resilient codes, which impose no minimum fragment length and instead bound the number of arbitrary cuts. They established a redundancy lower bound and offered a near-optimal construction~\cite{wang2024break,wang2026break}. Recently, an optimal code construction for this model is given in~\cite{wang2026optimal}. The model was later extended to permit fragment loss subject to a bound on the total lost length~\cite{wang2026break2}. The use of recovered information as positional reference in SR is inspired by the successive reconstruction principle in break-resilient coding, but the present work addresses stochastic cuts and average-error capacity.

\section{Preliminaries}
\label{sec:preliminaries}

\subsection{Notation}

Throughout the paper, all logarithms are base two.
We use~$\|$ to denote string concatenation, $\oplus$ to denote addition in~$\F$, and $\xrightarrow{\mathrm p}$ for convergence in probability. For a vector $V$ and an ordered index set $I$, the notation $V_I$ denotes the corresponding subvector.
 
\subsection{Channel Model}

Let
\[
  X=(X_0,X_1,\ldots,X_{N-1})\in\F^N
\]
be the transmitted codeword.
The stochastic torn-paper channel samples a cut vector
\[
  B=(B_0,B_1,\ldots,B_{N-2})\in\F^{N-1},
\]
whose entries are mutually independent and satisfy $\Pr(B_i=1)=p_N$.
The event $B_i=1$ places a cut between $X_i$ and $X_{i+1}$.
Consequently, the torn-paper channel produces
\[
  J=1+\sum_{i=0}^{N-2}B_i
\]
nonempty contiguous fragments.
Under the critical regime
\begin{equation}
  p_N\log N\longrightarrow\alpha\in(0,\infty),
  \label{eq:critical-regime}
\end{equation}
fragments have typical length~$\Theta(\log N)$.
The channel retains the orientation of each fragment and outputs the unordered multiset
\[
  \calY
  =
  \mltset{Y^{(1)},Y^{(2)},\ldots,Y^{(J)}}.
\]

\subsection{Codes and Error Criterion}

A blocklength-$N$ binary code consists of a finite message set $\calU_N$, an injective encoder $f_N:\calU_N\to\F^N$, and a decoder that maps every possible channel output to $\calU_N\cup\{\mathsf e\}$, where $\mathsf e$ denotes failure.
Its rate is
\[
  R_N\triangleq\frac{1}{N}\log|\calU_N|.
\]
The message $U$ is assumed to be uniform on $\calU_N$.
A rate $R$ is \textit{achievable under average error} if there is a sequence of binary codes with blocklengths tending to infinity for which $R_N\to R$ and $\Pr(\widehat U\neq U)\to0$.
The capacity is the supremum of all achievable rates.

Our construction is derived from a random public design $D_N$.
A realization $d_N$, known to both the encoder and decoder, specifies deterministic encoding and decoding maps.
For $u\in\calU_N$, define
\[
  P_e(u\mid d_N)
  \triangleq
  \Pr_B\{\widehat U\neq u\mid U=u,D_N=d_N\},
\]
where the probability is over the random cuts.
The average error probability of that realization is
\[
  \overline P_e(d_N)
  \triangleq
  \frac{1}{|\calU_N|}
  \sum_{u\in\calU_N}P_e(u\mid d_N).
\]
Section~\ref{sec:analysis} selects deterministic designs so that the resulting encoders are injective.

\section{Main Results}
\label{sec:results}

We first summarize the SR construction and then state its finite-track reliability and capacity guarantees.
The $N$ codeword coordinates are partitioned into $M\geq2$ interleaved tracks according to their residues modulo $M$.
The unshifted bootstrap track, indexed by $0$, is all zero, while tracks $1,\ldots,M-1$ carry disjoint message segments.
Let $n_j$ be the length of track $j$ and fix rates $r_1,\ldots,r_{M-1}$.
For $1\leq j<M$, the quantity $k_j\triangleq\floor{n_jr_j}$ is the number of message bits assigned to track $j$, and $G_j\in\F^{k_j\times n_j}$ is its generator matrix.
The tracks are interleaved and then shifted by a public random shift vector $Z\in\F^N$.
Thus, the public design is
\[
  D_N=(Z,G_1,\ldots,G_{M-1}).
\]
All entries of $Z,G_1,\ldots,G_{M-1}$ are mutually independent and distributed as $\operatorname{Bernoulli}(1/2)$.

Decoding proceeds successively.
At stage $j$, the decoder uses the already recovered tracks $0,\ldots,j-1$ as a positional reference to localize eligible fragments and recover track $j$.
Once recovered, track $j$ joins the reference used in later stages.
The reference therefore grows throughout decoding, allowing the fragment-length threshold to decrease.

\begin{definition}[Long-fragment coverage]
For $t\geq0$, define
\begin{equation}
  V_\alpha(t)\triangleq(1+\alpha t)e^{-\alpha t}.
  \label{eq:coverage-function}
\end{equation}
By~\cite[Lemma~6]{shomorony2021torn}, $V_\alpha(t)$ is the asymptotic fraction of coordinates covered by fragments of length at least $\ceil{t\log N}$.
\end{definition}

\begin{theorem}[Finite-track reliability]
\label{thm:decoding}
Fix $M\geq2$, $\xi>0$, and rates $r_1,\ldots,r_{M-1}$ satisfying
\begin{equation}
  0\leq r_j<
  V_\alpha\left((1+2\xi)\frac{M}{j}\right),
  \qquad 1\leq j<M.
  \label{eq:track-rate-condition}
\end{equation}
For the encoder and decoder in Section~\ref{sec:construction}, under~\eqref{eq:critical-regime},
\begin{equation}
  \epsilon_N
  \triangleq
  \max_{u\in\calU_N}
  \mathrm E_{D_N}\left[P_e(u\mid D_N)\right]
  \longrightarrow0.
  \label{eq:ensemble-error}
\end{equation}
\end{theorem}

The threshold in~\eqref{eq:track-rate-condition} has an intuitive interpretation. At stage $j$, a fragment of length $\ell$ contains approximately $j\ell/M$ symbols from recovered tracks.
Approximately $\log N$ independent reference symbols are needed to distinguish among $O(N)$ possible starting positions, which motivates a fragment threshold of order $(M/j)\log N$.
The factor $1+2\xi$ supplies slack.
Accordingly, stage $j$ provides approximately
\[
  n_jV_\alpha\left((1+2\xi)\frac{M}{j}\right)
\]
new bit observations. The strict rate inequality ensures that, with probability approaching one, the number of bit observations exceeds the number of message bits $k_j$.

\begin{theorem}[Capacity achievability]
\label{thm:capacity-achievability}
For $M\geq2$, define
\begin{equation}
  R_M(\alpha)
  \triangleq
  \frac{1}{M}\sum_{j=1}^{M-1}
  \left(1+\frac{\alpha M}{j}\right)e^{-\alpha M/j}.
  \label{eq:finite-M-rate}
\end{equation}
For every $0\leq R<R_M(\alpha)$, there exist $\xi>0$ and $r_1,\ldots,r_{M-1}$ satisfying~\eqref{eq:track-rate-condition} and
\[
  \frac{1}{M}\sum_{j=1}^{M-1}r_j=R.
\]
Moreover, there is a sequence of deterministic public designs $(d_N)$ for which the associated encoders are injective and
\[
  \frac{\sum_{j=1}^{M-1}\floor{n_jr_j}}{N}
  \longrightarrow R,
  \qquad
  \overline P_e(d_N)\longrightarrow0.
\]
Finally,
\begin{equation}
  0\leq e^{-\alpha}-R_M(\alpha)
  \leq\frac{(1+\alpha)e^{-\alpha}}{M}.
  \label{eq:finite-M-gap}
\end{equation}
Consequently, the SR construction is capacity-achieving under average error.
\end{theorem}

\section{The Successive-Revelation Construction}
\label{sec:construction}

\subsection{Encoder}

Fix $M\geq2$, $\xi>0$, and track rates $r_1,\ldots,r_{M-1}\in[0,1)$.
Let $N\geq M$ and write
\[
  N=qM+s,
  \qquad 0\leq s<M.
\]
The coordinates are partitioned into tracks according to their residues modulo $M$. For $0\leq j<M$, track $j$ has length
\[
  n_j
  =
  \begin{cases}
    q+1, & 0\leq j<s,\\
    q,   & s\leq j<M,
  \end{cases}
\]
and coordinate set
\[
  T_j
  \triangleq
  \{j,j+M,\ldots,j+(n_j-1)M\}.
\]

For $1\leq j<M$, set $k_j\triangleq\floor{n_jr_j}$ and
\[
  K\triangleq\sum_{j=1}^{M-1}k_j.
\]
The message set is $\calU_N=\F^K$. The encoder partitions
\[
  U
  =
  U^{(1)}\|U^{(2)}\|\cdots\|U^{(M-1)},
  \qquad
  U^{(j)}\in\F^{k_j}.
\]

The public design consists of $Z\in\F^N$ and matrices
\[
  G_j\in\F^{k_j\times n_j},
  \qquad 1\leq j<M.
\]
Their entries are mutually independent, and each is distributed as $\operatorname{Bernoulli}(1/2)$.

The encoder forms
\begin{equation}
  C^{(j)}
  =
  \begin{cases}
    0^{n_0}, & j=0,\\
    U^{(j)}G_j, & 1\leq j<M,
  \end{cases}
  \label{eq:encoded-tracks}
\end{equation}
and interleaves the tracks into
\begin{equation}
  C=(C_0,C_1,\ldots,C_{N-1})\in\F^N,
  \qquad
  C_{j+aM}=C^{(j)}_a,
  \label{eq:interleaving}
\end{equation}
for $0\leq j<M$ and $0\leq a<n_j$.
It then outputs
\begin{equation}
  X=C\oplus Z.
  \label{eq:shifted-codeword}
\end{equation}

\subsection{Decoder}

The decoder proceeds through stages $j=1,\ldots,M-1$.
Define
\[
  S_j\triangleq T_0\cup T_1\cup\cdots\cup T_{j-1}.
\]
At the beginning of stage $j$, provided all preceding stages have succeeded, reconstructed values $\widehat x_i$ are available for every $i\in S_j$. The bootstrap track is known from the public design, so initially
\[
  \widehat X_{aM}=Z^{(0)}_a,
  \qquad 0\leq a<n_0.
\]

The decoder maintains an unlocalized multiset $\mathcal R$, initially equal to $\calY$, and a cache $\mathcal P$ of localized fragment--position pairs, initially empty. At stage $j$, it uses the threshold
\begin{equation}
  L_{j,N}
  \triangleq
  \ceil{(1+2\xi)\frac{M}{j}\log N}.
  \label{eq:stage-fragment-threshold}
\end{equation}
For each fragment $Y=(Y_0,\ldots,Y_{\ell-1})\in\mathcal R$ with $\ell\geq L_{j,N}$, the decoder performs the following localization test. For a candidate starting coordinate $i\in\{0,\ldots,N-\ell\}$, let
\[
  Q_j(i,\ell)
  \triangleq
  \{r\in\{0,\ldots,\ell-1\}:i+r\in S_j\}.
\]
The positions consistent with the current reference are
\small\begin{equation} \label{eq:candidate-set}
  \begin{split}
     &\Lambda_j(Y)
      \triangleq\\
      &\{
        i\in\{0,\ldots,N-\ell\}:
        Y_r=\widehat X_{i+r}
        \text{ for every }r\in Q_j(i,\ell)
      \}.
  \end{split}
\end{equation}\normalsize
If $\Lambda_j(Y)=\{i\}$, the decoder records the pair $(Y,i)$ in $\mathcal P$ and removes that copy of $Y$ from $\mathcal R$. If the candidate set is not a singleton, the fragment remains in $\mathcal R$ and may be reconsidered at a later stage. Every recorded placement remains in $\mathcal P$ for all subsequent stages.

At stage $j$, the decoder uses every pair $(Y,i)\in\mathcal P$ with $|Y|\geq L_{j,N}$. If the physical coordinate $j+aM$ lies in the recorded placement, the pair supplies
\[
  W^{(j)}_a\triangleq Y_{j+aM-i}.
\]
Let $O_j\subseteq\{0,\ldots,n_j-1\}$ be the set of track indices observed in this way. If two cached placements supply a value for the same index, the decoder declares failure.

For $0\leq j<M$, define the subsequence
\[
  Z^{(j)}
  \triangleq
  (Z_{j+aM})_{a=0}^{n_j-1}.
\]
After ordering the elements of $O_j$ increasingly, the decoder removes the public shift by computing
\[
  V^{(j)}_{O_j}
  \triangleq
  W^{(j)}_{O_j}\oplus Z^{(j)}_{O_j}.
\]
Let $G_{j,O_j}$ be the submatrix of $G_j$ consisting of the columns indexed by $O_j$.
The decoder solves
\begin{equation}
  \widehat U^{(j)}G_{j,O_j}
  =
  V^{(j)}_{O_j}.
  \label{eq:track-decoding-system}
\end{equation}
If this system has a unique solution, the decoder reconstructs
\[
  \widehat C^{(j)}=\widehat U^{(j)}G_j,
  \qquad
  \widehat X^{(j)}=\widehat C^{(j)}\oplus Z^{(j)},
\]
and stores
\[
  \widehat X_{j+aM}\triangleq\widehat X^{(j)}_a,
  \qquad 0\leq a<n_j.
\]
Otherwise, it declares failure. After all stages have succeeded, the decoder outputs
\[
  \widehat U
  =
  \widehat U^{(1)}\|\widehat U^{(2)}
  \|\cdots\|\widehat U^{(M-1)}.
\]
If any stage fails, it outputs $\mathsf e$.

\section{Reliability and Rate Analysis}
\label{sec:analysis}

This section proves the main results in Section~\ref{sec:results} using four ingredients. First, we show that long fragments cover a predictable fraction of every periodic track. Second, the random shift makes an incorrect fragment placement unlikely to match the revealed tracks. Third, after localization, the observed submatrix of each random generator matrix has full row rank with high probability, enabling recovery of the corresponding track. Finally, we show that the finite-track rates form a Riemann sum that approaches the channel capacity.

We use the decoder's persistent cache throughout the analysis. When a fragment is uniquely localized, it is transferred from the unlocalized multiset~$\mathcal R$ to the cache~$\mathcal P$ together with its starting coordinate.
Its recorded placement remains available at every subsequent stage. A fragment that cannot yet be localized remains in~$\mathcal R$ and may be reconsidered after more tracks have been recovered. Thus, at stage~$j$, the set~$O_j$ includes observations from every eligible fragment localized at or before that stage.

\subsection{Random Shift and Long-Fragment Coverage}

We begin with a useful consequence of the random shift vector.

\begin{lemma}[Random-shift independence]
\label{lem:random-shift}
Fix a message~$u\in\calU_N$.
Under the random public design, the transmitted word~$X$ is uniform on~$\F^N$ and is independent of the generator matrices~$G_1,\ldots,G_{M-1}$.
\end{lemma}

\begin{proof}
Let~$G=(G_1,\ldots,G_{M-1})$.
For fixed~$u$ and~$G=g$, the unshifted word~$C(u,g)$ is deterministic, whereas~$Z$ is uniform on~$\F^N$.
Consequently, for every~$x\in\F^N$, the conditional probability
\begin{equation}
  \Pr\{X=x\mid G=g,U=u\}
  =
  \Pr\{Z=x\oplus C(u,g)\}
  =2^{-N}
\end{equation}
does not depend on~$g$. Therefore, the coordinates of~$X$ are mutually independent Bernoulli$(1/2)$ random variables, and~$X$ is independent of~$G$.
\end{proof}

For each channel-fragment occurrence~$Y$ (with identical copies distinguished only for analysis), let~$I(Y)\subseteq\{0,\ldots,N-1\}$ denote its true coordinate interval; this interval is not supplied to the decoder. For fixed~$t>0$, define the total long-fragment coverage and its track-$a$ counterpart by
\begin{align}
  H_N(t)
  &\triangleq
  \sum_{Y\in\calY}|Y|\,
  \mathbf 1_{\{|Y|\geq \ceil{t\log N}\}},
  \label{eq:total-long-coverage}\\
  H_{a,N}(t)
  &\triangleq
  \sum_{Y\in\calY}|I(Y)\cap T_a|\,
  \mathbf 1_{\{|Y|\geq \ceil{t\log N}\}},
  \qquad 0\leq a<M.
  \label{eq:track-long-coverage}
\end{align}

\begin{lemma}[Long-fragment coverage on each track]
\label{lem:track-coverage}
For every fixed~$t>0$ and~$a\in\{0,\ldots,M-1\}$,
\begin{equation}
  \frac{H_N(t)}{N}
  \xrightarrow{\mathrm p}
  V_\alpha(t)
  \qquad\text{and}\qquad
  \frac{H_{a,N}(t)}{n_a}
  \xrightarrow{\mathrm p}
  V_\alpha(t).
  \label{eq:track-coverage-limit}
\end{equation}
Moreover, the number~$J$ of fragments satisfies~$J/N\xrightarrow{\mathrm p}0$.
\end{lemma}

\begin{proof}
The first limit is the standard coverage lemma for the torn-paper channel~\cite[Lemma~6]{shomorony2021torn}. We briefly recall why its limit has the stated form.
If~$L$ is geometric with parameter~$p_N$ on~$\{1,2,\ldots\}$ and~$h_N=\ceil{t\log N}$, then the memoryless property gives
\begin{equation}
  p_N\,\mathrm E[L\mathbf 1_{\{L\geq h_N\}}]
  =
  (1+p_N(h_N-1))(1-p_N)^{h_N-1}.
  \label{eq:size-biased-geometric-tail}
\end{equation}
Since~$p_N\log N\to\alpha>0$, we have $p_N\sim\alpha/\log N$ and hence $Np_N^2\to\infty$. The right-hand side of~\eqref{eq:size-biased-geometric-tail} converges to $(1+\alpha t)e^{-\alpha t}=V_\alpha(t)$.
The cited coverage lemma shows that the realized coverage concentrates around this value. It also accounts for the terminal fragment, whose possible truncation has no asymptotic effect.

It remains to pass from total coverage to one periodic track.
For every integer interval~$I$,
\begin{equation}
  \left||I\cap T_a|-\frac{|I|}{M}\right|\leq 1.
  \label{eq:periodic-track-discrepancy}
\end{equation}
Summing this inequality over the long fragments yields
\begin{equation}
  \left|H_{a,N}(t)-\frac{H_N(t)}{M}\right|
  \leq J.
  \label{eq:track-total-discrepancy}
\end{equation}
Furthermore,
\begin{equation}
  \mathrm E\left[\frac{J}{N}\right]
  =\frac{1+(N-1)p_N}{N}
  \longrightarrow0,
\end{equation}
so~$J/N\xrightarrow{\mathrm p}0$ by Markov's inequality.
Finally,~$n_a=N/M+O(1)$. Dividing~\eqref{eq:track-total-discrepancy} by~$n_a$ and using the first limit proves the second.
\end{proof}

\subsection{Localization from previously recovered tracks}

Fix a stage~$j\in\{1,\ldots,M-1\}$.
Define an ideal stage-$j$ localization rule obtained by supplying the true values $(X_i)_{i\in S_j}$ as its positional reference. The induction in the proof of Theorem~\ref{thm:decoding} will couple this auxiliary rule to the actual decoder whenever the preceding stages have succeeded.
Set
\begin{equation}
  t_j\triangleq(1+2\xi)\frac{M}{j},
  \qquad
  v_j\triangleq V_\alpha(t_j).
  \label{eq:tj-vj}
\end{equation}

Call a fragment whose true starting coordinate is~$s$ \emph{ambiguous at stage~$j$} if some false starting coordinate~$i\neq s$ belongs to its candidate set~$\Lambda_j(Y)$. The next lemma shows that ambiguous fragments cover only a vanishing fraction of the codeword.

\begin{lemma}[Vanishing loss from ambiguous fragments]
\label{lem:localization}
Fix a message~$u$ and use the ideal stage-$j$ positional reference.
Let~$W_{j,N}$ be the total number of codeword coordinates contained in ambiguous fragments of length at least~$L_{j,N}$.
Then
\begin{equation}
  \mathrm E_{D_N,B}[W_{j,N}\mid U=u]
  \leq 2^jN^{1-2\xi},
  \qquad
  \frac{W_{j,N}}{N}
  \xrightarrow{\mathrm p}0.
  \label{eq:ambiguous-coverage-bound}
\end{equation}
If~$O_j^\star$ denotes the observation set obtained with the correct reference and persistent fragment reuse, and~$m_j^\star\triangleq|O_j^\star|$, then
\begin{equation}
  \frac{m_j^\star}{n_j}
  \xrightarrow{\mathrm p}
  v_j.
  \label{eq:observation-set-limit}
\end{equation}
Every fragment that is uniquely localized under the correct reference is localized to its true starting coordinate.
\end{lemma}

\begin{proof}
Condition on the cut vector, and consider a particular fragment with true start~$s$ and length~$\ell\geq L_{j,N}$.
For a false candidate~$i\neq s$, consistency requires
\begin{equation}
  X_{s+r}=X_{i+r}
  \qquad\text{for every }r\in Q_j(i,\ell).
  \label{eq:false-placement-equalities}
\end{equation}
Every block of~$M$ consecutive coordinates contains exactly~$j$ coordinates of~$S_j$.
Therefore, with~$q\triangleq|Q_j(i,\ell)|$,
\begin{equation}
  q
  \geq j\floor{\frac{\ell}{M}}
  \geq (1+2\xi)\log N-j.
  \label{eq:number-reference-comparisons}
\end{equation}

By Lemma~\ref{lem:random-shift}, the coordinates of~$X$ are independent fair bits. To count the independent constraints in~\eqref{eq:false-placement-equalities}, form a graph whose vertices are codeword coordinates and whose edges are the pairs~$\{s+r,i+r\}$ for~$r\in Q_j(i,\ell)$. All edges have the same nonzero displacement~$i-s$. Within each residue class modulo~$|i-s|$, the graph is a subgraph of a path, so it contains no cycle. A forest with $q$ edges imposes $q$ independent equality constraints on independent fair binary variables.
Therefore,
\begin{equation}
  \Pr_{D_N}\{i\in\Lambda_j(Y)\mid B,U=u\}
  =2^{-q}
  \leq 2^jN^{-(1+2\xi)}.
  \label{eq:false-candidate-probability}
\end{equation}
There are fewer than~$N$ false candidates.
A union bound therefore gives
\begin{equation}
  \Pr_{D_N}\{Y\text{ is ambiguous at stage }j\mid B,U=u\}
  \leq 2^jN^{-2\xi}.
  \label{eq:fragment-ambiguity-probability}
\end{equation}
The fragments are disjoint and their lengths sum to~$N$.
After multiplying~\eqref{eq:fragment-ambiguity-probability} by fragment lengths and summing, we obtain
\begin{equation}
  \mathrm E_{D_N}[W_{j,N}\mid B,U=u]
  \leq 2^jN^{-2\xi}
  \sum_{Y\in\calY}|Y|
  =2^jN^{1-2\xi}.
\end{equation}
This proves the expectation bound, and Markov's inequality proves the convergence in~\eqref{eq:ambiguous-coverage-bound}.

The true start~$s$ always belongs to~$\Lambda_j(Y)$ when the reference is correct. Hence, if the candidate set is a singleton, its element must be~$s$. Correctly placed fragments are disjoint, so they never give conflicting observations at the same coordinate. All track-$j$ coordinates lying in long, nonambiguous fragments belong to~$O_j^\star$.
Consequently,
\begin{equation}
  H_{j,N}(t_j)-W_{j,N}
  \leq m_j^\star
  \leq H_{j,N}(t_j).
  \label{eq:observations-sandwich}
\end{equation}
Divide by~$n_j$, use~$n_j=N/M+O(1)$, and apply Lemma~\ref{lem:track-coverage} and~\eqref{eq:ambiguous-coverage-bound} to obtain~\eqref{eq:observation-set-limit}.
\end{proof}

The last ingredient is a standard rank estimate.

\begin{lemma}[Rank of a random observed submatrix]
\label{lem:random-rank}
Let~$G\in\F^{k\times n}$ have mutually independent Bernoulli$(1/2)$ entries, and let~$O\subseteq\{0,\ldots,n-1\}$ be independent of~$G$.
Let $G_O$ denote the submatrix consisting of the columns indexed by $O$.
Conditional on~$O$ and with~$m\triangleq|O|\geq k$,
\begin{equation}
  \Pr\{\rank(G_O)<k\mid O\}
  <2^{k-m}.
  \label{eq:random-rank-bound}
\end{equation}
\end{lemma}

\begin{proof}
The matrix~$G_O$ fails to have row rank~$k$ only if some nonzero row vector~$a\in\F^k$ satisfies~$aG_O=0$.
For each fixed nonzero~$a$, its inner product with every selected column is an independent fair bit, so this event has probability~$2^{-m}$.
A union bound over the~$2^k-1$ nonzero vectors gives
\begin{equation}
  \Pr\{\rank(G_O)<k\mid O\}
  \leq(2^k-1)2^{-m}<2^{k-m}.
\end{equation}
\end{proof}

\subsection{Proof of finite-track reliability}

\begin{proof}[Proof of Theorem~\ref{thm:decoding}]
Fix~$u\in\calU_N$.
For stage~$j$, let~$t_j$ and~$v_j$ be as in~\eqref{eq:tj-vj}, and choose a constant~$\eta_j>0$ such that
\begin{equation}
  r_j+2\eta_j<v_j.
  \label{eq:track-rate-slack}
\end{equation}
By Lemma~\ref{lem:localization},
\begin{equation}
  \Pr\{m_j^\star<(r_j+\eta_j)n_j\}
  \longrightarrow0.
  \label{eq:enough-observations}
\end{equation}

The ideal observation set~$O_j^\star$ is a function of~$X$ and the cut vector~$B$. Lemma~\ref{lem:random-shift} shows that it is independent of~$G_j$. On the event that~$m_j^\star\geq(r_j+\eta_j)n_j$, we have
\begin{equation}
  m_j^\star-k_j
  \geq \eta_jn_j,
\end{equation}
because~$k_j=\floor{n_jr_j}\leq n_jr_j$.
Lemma~\ref{lem:random-rank} now gives
\begin{equation}
  \Pr\{\rank(G_{j,O_j^\star})<k_j,
  \ m_j^\star\geq(r_j+\eta_j)n_j\}
  \leq 2^{-\eta_jn_j}.
  \label{eq:stage-rank-failure}
\end{equation}
Since~$M$ is fixed, a union bound over~$j=1,\ldots,M-1$ shows that, with probability tending to one, every stage has enough observations and every corresponding observed submatrix has full row rank.

We finish by induction over the decoding stages. Before stage~$1$, track~$0$ is known exactly. Suppose tracks~$0,\ldots,j-1$ have been reconstructed correctly. Then the decoder's actual reference equals the correct reference used to define~$O_j^\star$. Because localized placements remain in the persistent cache, its actual observation set is therefore exactly~$O_j^\star$.
Lemma~\ref{lem:localization} implies that each uniquely localized fragment is placed correctly, so the right-hand side of~\eqref{eq:track-decoding-system} equals
\begin{equation}
  V^{(j)}_{O_j^\star}
  =U^{(j)}G_{j,O_j^\star}.
\end{equation}
Full row rank makes~$U^{(j)}$ the unique solution.
The decoder therefore reconstructs track~$j$ exactly and adds its true transmitted values to the reference.
This completes the induction.

None of the preceding probability bounds depends on the value of~$u$.
Therefore,
\begin{equation}
  \max_{u\in\calU_N}
  \Pr_{D_N,B}\{\hat U\neq u\mid U=u\}
  \longrightarrow0.
\end{equation}
Using
$\Pr_{D_N,B}\{\hat U\neq u\mid U=u\}
=\mathrm E_{D_N}[P_e(u\mid D_N)]$
proves~\eqref{eq:ensemble-error}.
\end{proof}

\subsection{Proof of capacity achievability}

\begin{proof}[Proof of Theorem~\ref{thm:capacity-achievability}]
For~$\xi\geq0$, define
\begin{equation}
  A_M(\xi)
  \triangleq\frac1M\sum_{j=1}^{M-1}
  V_\alpha\left((1+2\xi)\frac{M}{j}\right).
  \label{eq:margin-rate-sum}
\end{equation}
This function is continuous at~$\xi=0$, and~$A_M(0)=R_M(\alpha)$.
Hence, if~$0\leq R<R_M(\alpha)$, there is a sufficiently small~$\xi>0$ such that~$R<A_M(\xi)$.
For~$R>0$, set
\begin{equation}
  r_j
  \triangleq\frac{R}{A_M(\xi)}
  V_\alpha\left((1+2\xi)\frac{M}{j}\right),
  \qquad 1\leq j<M.
  \label{eq:capacity-track-rate-choice}
\end{equation}
For~$R=0$, set all~$r_j=0$.
These choices satisfy the strict track-rate inequalities and
\begin{equation}
  \frac1M\sum_{j=1}^{M-1}r_j=R.
\end{equation}

Because~$n_j=N/M+O(1)$ and the number of tracks is fixed,
\begin{align}
  \frac{K}{N}
  &=\frac1N\sum_{j=1}^{M-1}\floor{n_jr_j}\\
  &\longrightarrow
  \frac1M\sum_{j=1}^{M-1}r_j
  =R.
  \label{eq:deterministic-code-rate}
\end{align}

It remains to select valid deterministic public designs. Theorem~\ref{thm:decoding} gives
\begin{equation}
  \mathrm E_{D_N}[\overline P_e(D_N)]
  =\frac1{|\calU_N|}\sum_{u\in\calU_N}
  \mathrm E_{D_N}[P_e(u\mid D_N)]
  \leq\epsilon_N
  \longrightarrow0.
  \label{eq:average-design-error}
\end{equation}
Let~$\mathcal F_N$ be the event that every~$G_j$ has row rank~$k_j$. The same argument as in Lemma~\ref{lem:random-rank}, now using all~$n_j$ columns, gives
\begin{equation}
  \Pr(\mathcal F_N^c)
  \leq\sum_{j=1}^{M-1}2^{k_j-n_j}
  \longrightarrow0,
  \label{eq:full-generator-probability}
\end{equation}
where the last limit follows from~$r_j<1$.
Thus
\begin{equation}
  \mathrm E_{D_N}[\overline P_e(D_N)\mid\mathcal F_N]
  \leq\frac{\epsilon_N}{\Pr(\mathcal F_N)}
  \longrightarrow0.
\end{equation}
For every sufficiently large~$N$, choose a realization~$d_N\in\mathcal F_N$ whose average error is no larger than this conditional mean. Every component map~$U^{(j)}\mapsto U^{(j)}G_j$ is injective under~$\mathcal F_N$; interleaving and adding the fixed shift preserve injectivity.
Hence~$d_N$ specifies a valid deterministic code, and~$\overline P_e(d_N)\to0$.

We finally prove the finite-$M$ gap.
Define, for~$0<x\leq1$,
\begin{equation}
  h(x)\triangleq\left(1+\frac{\alpha}{x}\right)e^{-\alpha/x},
  \qquad h(0)\triangleq0.
\end{equation}
Both $xe^{-\alpha/x}$ and $h(x)$ tend to zero as $x\downarrow0$, so these endpoint extensions are continuous.
Since
\begin{equation}
  \frac{\mathrm d}{\mathrm dx}
  \left(xe^{-\alpha/x}\right)=h(x),
\end{equation}
we have
\begin{equation}
  \int_0^1h(x)\,\mathrm dx=e^{-\alpha}.
  \label{eq:capacity-integral}
\end{equation}
Moreover,
\begin{equation}
  h'(x)=\frac{\alpha^2}{x^3}e^{-\alpha/x}>0,
\end{equation}
so~$h$ is increasing, and
\begin{equation}
  R_M(\alpha)
  =\frac1M\sum_{j=0}^{M-1}h\left(\frac jM\right)
\end{equation}
is the left Riemann sum for~\eqref{eq:capacity-integral}.
Therefore~$R_M(\alpha)\leq e^{-\alpha}$.
For each subinterval, monotonicity also gives
\begin{align}
  0
  &\leq
  \int_{j/M}^{(j+1)/M}h(x)\,\mathrm dx
  -\frac1M h\left(\frac jM\right)\\
  &\leq
  \frac1M\left[
  h\left(\frac{j+1}{M}\right)
  -h\left(\frac jM\right)
  \right].
\end{align}
Summing over~$j=0,\ldots,M-1$ telescopes and yields
\begin{equation}
  0\leq e^{-\alpha}-R_M(\alpha)
  \leq\frac{h(1)-h(0)}{M}
  =\frac{(1+\alpha)e^{-\alpha}}{M}.
\end{equation}

Given any~$R<e^{-\alpha}$, the preceding bound allows us to choose a fixed $M$ large enough that~$R<R_M(\alpha)$.
The first part of the proof then supplies a deterministic SR code sequence of rate~$R$ and vanishing average error.
Together with the converse in~\cite{shomorony2021torn}, this proves capacity achievability.
\end{proof}

\section{Conclusion}
\label{sec:conclusion}

We introduced successive revelation, a capacity-achieving coding method for the binary stochastic torn-paper channel under the average-error criterion. The construction begins with a known bootstrap track and incorporates each recovered payload track into the positional reference.
As the reference grows, the decoder can localize progressively shorter fragments. For a fixed number of tracks, the resulting rate is an average of long-fragment coverage values, and its gap to the TPC capacity is $O(1/M)$.
Hence, every rate below $e^{-\alpha}$ is achieved by a finite-track construction with vanishing average decoding error probability.

\printbibliography

@article{shomorony2021torn,
  author={Shomorony, Ilan and Vahid, Alireza},
  journal={IEEE Transactions on Information Theory}, 
  title={Torn-Paper Coding}, 
  year={2021},
  volume={67},
  number={12},
  pages={7904--7913}
}

@inproceedings{shomorony2020communicating,
  title={Communicating over the torn-paper channel},
  author={Shomorony, Ilan and Vahid, Alireza},
  booktitle={2020 IEEE Global Communications Conference (GLOBECOM)},
  pages={1--6},
  year={2020},
  organization={IEEE}
}

@article{liu2026improved,
  title={Improved Torn Paper Coding via Local Alignment},
  author={Liu, Junsheng and Raviv, Netanel},
  journal={arXiv preprint arXiv:2605.23076},
  year={2026}
}

@inproceedings{ravi2021capacity,
  title={Capacity of the torn paper channel with lost pieces},
  author={Ravi, Aditya Narayan and Vahid, Alireza and Shomorony, Ilan},
  booktitle={2021 IEEE International Symposium on Information Theory (ISIT)},
  pages={1937--1942},
  year={2021},
  organization={IEEE}
}

@article{ravi2026recovering,
  title={Recovering a message from an incomplete set of noisy fragments},
  author={Ravi, Aditya Narayan and Vahid, Alireza and Shomorony, Ilan},
  journal={IEEE Transactions on Information Theory},
  year={2026},
  publisher={IEEE}
}

@article{bar2023adversarial,
  title={Adversarial torn-paper codes},
  author={Bar-Lev, Daniella and Marcovich, Sagi and Yaakobi, Eitan and Yehezkeally, Yonatan},
  journal={IEEE Transactions on Information Theory},
  volume={69},
  number={10},
  pages={6414--6427},
  year={2023},
  publisher={IEEE}
}

@article{walter2026coding,
  title={Coding Schemes for the Noisy Torn Paper Channel},
  author={Walter, Frederik and Abu-Sini, Maria and Weinhardt, Nils and Wachter-Zeh, Antonia},
  journal={arXiv preprint arXiv:2601.11501},
  year={2026}
}

@inproceedings{wang2024break,
  title={Break-resilient codes for forensic {3D} fingerprinting},
  author={Wang, Canran and Sima, Jin and Raviv, Netanel},
  booktitle={2024 IEEE International Symposium on Information Theory (ISIT)},
  pages={3148--3153},
  year={2024},
  organization={IEEE}
}

@article{wang2026break,
  title={Break-Resilient Codes},
  author={Wang, Canran and Sima, Jin and Raviv, Netanel},
  journal={IEEE Transactions on Information Theory},
  year={2026},
  publisher={IEEE}
}

@article{wang2026break2,
  title={Break-resilient codes with loss tolerance},
  author={Wang, Canran and Liwang, Minghui and Raviv, Netanel},
  journal={arXiv preprint arXiv:2601.14623},
  year={2026}
}

@article{wang2026optimal,
  title={Optimal Break-Resilient Codes},
  author={Wang, Canran},
  journal={arXiv preprint arXiv:2607.19673},
  year={2026}
}

@article{nassirpour2020embedded,
  title={Embedded codes for reassembling non-overlapping random {DNA} fragments},
  author={Nassirpour, Sajjad and Vahid, Alireza},
  journal={IEEE Transactions on Molecular, Biological, and Multi-Scale Communications},
  volume={7},
  number={1},
  pages={40--50},
  year={2021},
  publisher={IEEE}
}

@article{nassirpour2023dna,
  title={DNA merge-sort: A family of nested varshamov-tenengolts reassembly codes for out-of-order media},
  author={Nassirpour, Sajjad and Shomorony, Ilan and Vahid, Alireza},
  journal={IEEE Transactions on Communications},
  volume={72},
  number={3},
  pages={1303--1317},
  year={2023},
  publisher={IEEE}
}

@article{jiao2025efficient,
  title={Efficient Nested Hash Reassembly Codes for Torn Paper Channels},
  author={Jiao, Xiaopeng and Jiao, Botao and Mu, Jianjun and Han, Hui},
  journal={IEEE Transactions on Communications},
  volume={73},
  number={10},
  pages={8607--8622},
  year={2025},
  publisher={IEEE}
}

\end{document}